\documentclass{article}
\usepackage{excstyle}
\usepackage{amsthm} 
\usepackage{verbatim}

\title{\mytitle}

\author{Jean-Guillaume Dumas\thanks{Laboratoire J. Kuntzmann,
    Universit\'e de Grenoble. 51, rue des Math\'ematiques, umr CNRS
    5224, bp 53X, F38041 Grenoble, France,
    \href{mailto:Jean-Guillaume.Dumas@imag.fr,Dominique.Duval@imag.fr,Burak.Ekici@imag.fr}{\{Jean-Guillaume.Dumas,Dominique.Duval,Burak.Ekici\}@imag.fr}.}
  \and Dominique Duval\footnotemark[1]
  \and
  Burak Ekici\footnotemark[1]
  \and
  Jean-Claude~Reynaud\thanks{Reynaud Consulting (RC), \href{mailto:Jean-Claude.Reynaud@imag.fr}{Jean-Claude.Reynaud@imag.fr}.}
}

\theoremstyle{plain} 
\newtheorem{theorem}{Theorem}[section]

\newtheorem{lemma}[theorem]{Lemma}
\theoremstyle{definition}
\newtheorem{definition}[theorem]{Definition}
\theoremstyle{remark}
\newtheorem{remark}[theorem]{Remark}

\begin{document}
\maketitle

\begin{abstract} 
Exception handling is provided by most modern programming languages.
It allows to deal with anomalous or exceptional events 
which require special processing. 
In computer algebra, exception handling is an efficient way to implement 
the dynamic evaluation paradigm: 
for instance, in linear algebra, dynamic evaluation can be used 
for applying programs which have been written for matrices 
with coefficients in a field to matrices with coefficients in a ring. 
Thus, a proof system for computer algebra should include 
a treatement of exceptions, 
which must rely on a careful description of a semantics of exceptions. 
The categorical notion of monad 
can be used for formalizing the raising of exceptions: 
this has been proposed by Moggi and implemented in Haskell. 
In this paper, we provide a proof system for exceptions 
which involves both raising and handling, by extending Moggi's approach. 
Moreover, the core part of this proof system is dual 
to a proof system for side effects in imperative languages, 
which relies on the categorical notion of comonad. 
Both proof systems are implemented in the Coq proof assistant.  
\end{abstract}

\section{Introduction}

Using exceptions is an efficient way to simultaneously compute 
with dynamic evaluation in exact linear algebra. 
For instance, for computing the rank of a matrix, 
a program written for matrices with coefficients in a field 
can easily be modified by adding an exception treatment 
in order to be applied to a matrix with coefficients in a ring: 
the exception mechanism provides an automatic case distinction 
whenever it is required. 
 
The question that arises then is how to prove correctness of the new algorithm.
It would be nice to design proofs with two levels, as we designed the
algorithms: a first level without exceptions, then a second level taking the
exceptions into account.  
In this paper, we propose a proof system following this principle.

Exceptions form a \emph{computational effect}, 
in the sense that a syntactic expression $f:X\to Y$ 
is not always interpreted as a function $\deno{f}:\deno{X} \to \deno{Y}$
(where, as usual, the sets $\deno{X}$ and $\deno{Y}$
denote the interpretations of the types $X$ and $Y$). 
For instance a function which raises an exception 
can be interpreted as a function $\deno{f}:\deno{X} \to \deno{Y}+\Exc$ 
where $\Exc$ is the set of exceptions
and ``$+$'' denotes the disjoint union. 
In a programming language, exceptions usually differ from errors in
the sense that it is possible to recover from an exception while this
is impossible for an error; thus, exceptions have to be both raised
and handled.

The fundamental computational effect is the evolution of states in 
an imperative language, when seen from a functional point of view. 
There are several frameworks for 
combining functional and imperative programming: 
the effect systems classify the expressions according to 
the way they interact with the store~\cite{Lucassen88polymorphiceffect},
while the Kleisli category of the monad of states 
$(X\times \St)^\St$ (where $\St$ is the set of states) 
also provides a classification of expressions~\cite{Mo91};
indeed, both approaches are related~\cite{Wadler98themarriage}.
Lawvere theories were proposed for dealing with 
the operations and equations related to computational effects~\cite{PP02,HP07}.
Another related approach, based on the fact that the state is observed, 
relies on the classification of expressions provided by
the coKleisli category of comonad of states $X\times \St$ 
and its associated Kleisli-on-coKleisli category~\cite{DDFR12-state}.

The treatment of exceptions is another fundamental computational effect. 
It can be studied from the point of view of the monad of 
exceptions $X+\Exc$ (where $\Exc$ is the set of exceptions),
or with a Lawvere theory, however in these settings it is difficult to 
handle exceptions because this operation 
does not have the required algebraicity property~\cite{PP03,PP09}. 
This issue has been circumvented in~\cite{SM04} in order to get a Hoare logic 
for exceptions, 
in~\cite{Le06} by using both algebras and coalgebras. 
and in~\cite{PP09} by introducing handlers.
The formalization of exceptions can also be made from 
a coalgebraic point of view~\cite{Jacobs:2001:javaexc}. 
In this paper we extend Moggi's original approach: 
we use the classification of expressions provided by
the Kleisli category of the monad of exceptions 
and its associated coKleisli-on-Kleisli category;
moreover, we use the duality between states and exceptions 
discovered in~\cite{DDFR12-dual}.
However, it is not necessary to know about 
monads or comonads for reading this paper: 
the definitions and results are presented in an elementary way, 
in terms of equational theories.

In Section~\ref{sec:dynev} we give a motivating example 
for the use of exceptions as an efficient way to compute 
with dynamic evaluation in exact linear algebra.
 
Then in Section~\ref{sec:syn} we define the syntax of a simple language 
for dealing with exceptions.
 
The intended denotational semantics is described in Section~\ref{sec:den}:
we dissociate the core operations for switching 
between exceptions and non-exceptional values, on one side, 
from their encapsulation in the raising and handling operations,
on the other side.  
 
In Section~\ref{sec:deco} we add decorations to the syntax, 
in order to classify the expressions of the language according to 
their interaction with the exceptions.
Decorations extend the syntax much like compiler qualifiers or specifiers (e.g.,
like the \texttt{throw} annotation in C++ functions' signatures).
In addition, we also decorate the equations, 
which provides a proof system for dealing with this decorated syntax.
The main result of the paper is that this proof system 
is sound with respect to the semantics of exceptions 
(Theorem~\ref{theo:rules-sound}). 
A major property of this new proof system is that 
we can separate the verification of the proofs in two steps:
a first step checks properties of the programs  
whenever there is no effect, 
while a second step takes the effects into account via the decorations. 
Several properties of exceptions have been proven using this proof system,
and these proofs have been verified in the Coq proof assistant.

\section{Rank computations modulo composite numbers}\label{sec:dynev}

Rank algorithms play a fundamental role in computer algebra. 
For instance, computing homology groups of
simplicial complexes reduces to computing ranks and integer
Smith normal forms of the associated boundary
matrices~\cite{jgd:2003:GAP}.
One of the most efficient method for computing the Smith normal form
of such boundary matrices also reduces to computing ranks but
modulo the valence, a product of the primes involved in the Smith
form~\cite{jgd:2001:JSC}.
Now rank algorithms (mostly Gaussian elimination and Krylov based
methods) work well over fields (note that Gaussian elimination can be adapted
modulo powers of primes~\cite[\S 5.1]{jgd:2001:JSC}).
Modulo composite numbers, zero divisors arise. Gauss-Bareiss method
could be used but would require to know the determinant in advance,
with is more difficult than the valence.
The strategy used in \cite{jgd:2001:JSC} is to factor the valence, but
only partially (factoring is still not polynomial).
The large remaining composite factors will have very few zero divisors
and thus Gaussian elimination or Krylov methods will have very few
 
risks of falling upon them. Thus one can use dynamic evaluation: 
try to launch the rank algorithm modulo this composite number with large prime
factors and ``pretend'' to be working over a field~\cite{Duval:1989:evdyn}. 
In any case, if a zero divisor is encoutered, then the valence has
been further factored (in polynomial time!) and the algorithm can split
its computation independently modulo both factors. 
 
An effective algorithm design, here in C++, enabling this dynamic
evaluation with very little code modification, uses exceptions:
\begin{enumerate}
\item Add one exception at the arithmetic level, for signaling a
  tentative division by a zero divisor, see Fig.~\ref{fig:invmod}.
\begin{figure}[ht]
\begin{jgdfrsh}\input{invmodexc.cpp.tex}\end{jgdfrsh}
\caption{Throwing an exception upon division by a non unit}\label{fig:invmod}
\end{figure}
\item Catch this exception inside the rank algorithm and throw a new
  exception with the state of the rank iteration, see
  Fig.~\ref{fig:currentrank}
  (in our implementation the class \texttt{zmz} wraps integers modulo $m$, 
  the latter modulus being a \texttt{static} global variable).
\begin{figure}[ht]
\begin{jgdfrsh}\input{currentrank.cpp.tex}\end{jgdfrsh}
\caption{Re-throwing an exception to forward rank and iteration number
  information}\label{fig:currentrank}
\end{figure}
\item Then it is sufficient to wrap the rank algorithm with a
  dynamic evaluation, splitting the continuation modulo both factors, 
  see Fig.~\ref{fig:dyneval}.
\end{enumerate}
\begin{figure}[ht]
\begin{jgdfrsh}\input{dynevalgauss.cpp.tex}\end{jgdfrsh}
\caption{Recursive splitting wrapper around classical Gaussian elimination,
  packing a list of pairs of rank and associated modulus.}\label{fig:dyneval}
\end{figure}
The advantage of using exceptions over other software design is
twofold:

first, throwing an exception at the arithmetic level and not only on
  the inverse call in the rank algorithm allows to prevent that other
  unexpected divisions by zero divisors go unnoticed;
 
second, re-throwing a new exception in the rank algorithm allows to keep
  its specifications unchanged. It also enables to keep the
  modifications of rank algorithms to a minimum and to clearly separate
  normal behavior with primes from the exceptional handling of
  splitting composite moduli.

In the following, we propose a proof system with decorations, so that
proofs can easily be made in two steps: 
a first step without exceptions, that is, 
just preserving an initial proof of the rank algorithm;
then a second level taking the exceptions into account.

\section{Syntax for exceptions}\label{sec:syn}
The syntax for exceptions in computer languages depends on the language:  
the keywords for raising exceptions may be either 
\texttt{raise} or \texttt{throw}, 
and for handling exceptions they may be either 
\texttt{handle}, \texttt{try-with}, \texttt{try-except} or \texttt{try-catch},
for instance. 
In this paper we rather use \texttt{throw} and \texttt{try-catch}. 
 
The syntax for dealing with exceptions may be described in two parts: 
a basic part which deals with the basic data types 
and an exceptional part for raising and handling exceptions.

The \emph{basic} part of the syntax is a signature $\Sig_\basic$,
made of a \emph{types} (or  \emph{sorts}) and \emph{operations}.

The \emph{exceptional types} form a subset $\exctype$ 
of the set of types of $\Sig_\basic$. 
For instance in C++ any type (basic type or class) is an exceptional type, 
while in Java
 
there is a base class for exceptional types, 
such that the exceptional types are precisely the subtypes of this base class.

Now, we assume that some basic signature $\Sig_{\basic}$ 
and some set of exceptional types $\exctype$ have been chosen.

The signature $\Sig_\exc$ for exceptions is made of $\Sig_{\basic}$
 
together with the operations for raising and handling exceptions, as follows. 

\begin{definition}
\label{defi:syn-sig}

The \emph{signature for exceptions} $\Sig_\exc$ 
is made of $\Sig_{\basic}$ with 
 
the following operations: 
a \emph{raising} operation 
for each exceptional type $T$ and each type $Y$:
  $$ \throw{T}{Y} :T\to Y \;,$$ 
and a \emph{handling} operation 
for each $\Sig_\exc$-term $f:X\to Y$, 
each non-empty list of exceptional types $(T_i)_{1\leq i\leq n}$ 
and each family of $\Sig_\exc$-terms $(g_i:T_i\to Y)_{1\leq i\leq n}$:
  $$ \try{f}{\catchi} : X \to Y \;.$$ 
\end{definition}

An important, and somewhat surprising, feature of a language with exceptions 
is that all expressions in the language,
including the \emph{try-catch} expressions, propagate exceptions.
Indeed, if an exception is raised before some \emph{try-catch} expression 
is evaluated, this exception is propagated. 
In fact, the \emph{catch} block in a \emph{try-catch} expression 
may recover from exceptions which are raised inside the \emph{try} block, 
but the \emph{catch} block alone is not an expression of the language.  
 
This means that the operations for catching exceptions 
are \emph{private} operations: they are not part of the signature 
for exceptions. 
More precisely, the operations for raising and handling exceptions 
can be expressed in terms of a private \emph{empty type}
and two families of private operations: 
the \emph{tagging} operations for creating exceptions    
and the \emph{untagging} operations for catching them 
(inside the \emph{catch} block of any \emph{try-catch} expression).  
The tagging and untagging operations are called the \emph{core operations} 
for exceptions.
They are not part of $\Sig_\exc$, but the interpretations of the 
operations for raising and handling exceptions, which are part of $\Sig_\exc$, 
are defined in terms of the interpretations of the core operations.
The meaning of the core operations is given in Section~\ref{sec:den}.

\begin{definition}
\label{defi:syn-core}

Let $\Sig_\exc$ be the signature for exceptions.
The \emph{core} of $\Sig_\exc$ is the signature $\Sig_\core$ 
made of $\Sig_{\basic}$ with 
a type $\empt$ called the \emph{empty type} 
and two operations for each exceptional type $T$:  
$$ \tagg{T}:T\to\empt \;\;\mbox{ and }\;\; \untag{T}:\empt\to T $$ 
where $\tagg{T}$ is called 
the \emph{exception constructor} or the \emph{tagging} operation 
and $\untag{T}$ is called 
the \emph{exception recovery} or the \emph{untagging} operation.
\end{definition}

\section{Denotational semantics for exceptions}\label{sec:den}

In this Section we define a denotational semantics for exceptions 
which relies on the common semantics of exceptions in various languages, 
for instance in C++~\cite[Ch.~15]{cpp}, Java~\cite[Ch.~14]{java} 
or ML. 

The basic part of the syntax is interpreted in the usual way: 
each type $X$ is interpreted as a set~$\deno{X}$ 
and each operation $f:X\to Y$ of $\Sig_\basic$ as 
a function $\deno{f}:\deno{X} \to \deno{Y}$.

But, when $f:X\to Y$ in $\Sig_\exc$ is a raising or handling operation,
or when $f:X\to Y$ in $\Sig_\core$ is a tagging or untagging operation,
it is not interpreted as a function $\deno{f}:\deno{X} \to \deno{Y}$: 
this corresponds to the fact that the exceptions form a 
\emph{computational effect}.

The distinction between ordinary and exceptional values 
is discussed in Subsection~\ref{subsec:den-sem}.
Then, denotational semantics of raising and handling exceptions are
considered in Subsections~\ref{subsec:den-tag} and~\ref{subsec:den-untag},
respectively. 
We assume that some interpretation of $\Sig_{\basic}$ has been chosen.

\subsection{Ordinary values and exceptional values}\label{subsec:den-sem}
 
In order to express the denotational semantics of exceptions, 
a fundamental point is to distinguish between two kinds of values:
the ordinary (or non-exceptional) values and the exceptions.
It follows that the operations may be classified 
according to the way they may, or may not, interchange 
these two kinds of values: 
an ordinary value may be \emph{tagged} for constructing an exception, 
and later on the tag may be cleared in order to recover the value; 
then we say that the exception gets \emph{untagged}. 
 
Let $\Exc$ be a set, called the \emph{set of exceptions}.
 
For each set $A$, the set $A+\Exc$ is the disjoint union of $A$ and $\Exc$ 
and the canonical inclusions are denoted 
$\inn_A: A \to A+\Exc $ and $ \ina_A: \Exc \to A+\Exc$. 
For each functions $f:A\to B$ and $g:\Exc \to B$,
we denote by $\cotuple{f|g}:A+\Exc\to B$ the unique function 
such that $\cotuple{f|g}\circ \inn_A =f$ and $\cotuple{f|g}\circ \ina_A =g$.

\begin{definition}
\label{defi:exc} 
An element of $A+\Exc$ is 
an \emph{ordinary value} if it is in $A$ 
and an \emph{exceptional value} if it is in $\Exc$.
 
A function $\varphi:A+\Exc \to B+\Exc$: 
\begin{itemize}
 
\item \emph{raises an exception} 
if there is some $x\in A$ such that $\varphi(x)\in \Exc$.
\item \emph{recovers from an exception} 
if there is some $e\in \Exc$ such that $\varphi(e)\in B$. 
\item \emph{propagates exceptions} 
if $\varphi(e)=e$ for every $e\in \Exc$.
\end{itemize}
\end{definition}
 
Clearly, a function $\varphi:A+\Exc \to B+\Exc$ which propagates exceptions
may raise an exception, but cannot recover from an exception. 
Such a function $\varphi$ is characterized by its restriction 
$\varphi_{|_A}:A \to B+\Exc$, 
since its restriction on exceptions $\varphi_{|_\Exc}:\Exc \to B+\Exc$
is the inclusion $\ina_B$ of $\Exc$ in $B+\Exc$.

In the denotational semantics for exceptions, we will see that 
a term $f:X\to Y$ of $\Sig_\exc$ or $\Sig_\core$ 
may be interpreted either as 
a function $\deno{f}:\deno{X}\to \deno{Y}$ or as 
a function $\deno{f}:\deno{X}\to \deno{Y}+\Exc$ or as 
a function $\deno{f}:\deno{X}+\Exc\to \deno{Y}+\Exc$. 
However, in all cases, it is possible to convert $\deno{f}$ 
to a function from $\deno{X}+\Exc$ to $\deno{Y}+\Exc$,
as follows. 
 
\begin{definition}
\label{defi:upcast}
The \emph{upcasting conversions} are the following transformations:
\begin{itemize}
\item every function $\varphi:A\to B$ gives rise to 
$ \upcastone{\varphi}=\inn_{B}\circ\varphi:A\to B+\Exc $,
\item every function $\psi:A\to B+\Exc$ gives rise to 
$ \upcasttwo{\psi}=\cotuple{\psi|\ina_B}:A+\Exc \to B+\Exc $, 
which is equal to $\psi$ on $A$ and which propagates exceptions; 
\item it follows that every function $\varphi:A\to B$ gives rise to 
$\upcastboth{\varphi}=\upcasttwo{(\upcastone{\varphi})}
=\cotuple{\inn_{B}\circ\varphi|\ina_B}
= \varphi+\id_{\Exc}:A+\Exc \to B+\Exc$, 
which is equal to $\varphi$ on $A$ and which propagates exceptions.
\end{itemize}
\end{definition}
Since the upcasting conversions are \emph{safe} (i.e., injective), 
when there is no ambiguity the symbols $\upcastone$, $\upcasttwo$  
and $\upcastboth$ may be omitted. 
 
In this way, for each $f:X\to Y$ and $g:Y\to Z$, whatever their effects, 
we get $\deno{f}:\deno{X}+\Exc\to\deno{Y}+\Exc$
and $\deno{g}:\deno{Y}+\Exc\to\deno{Z}+\Exc$, which can be composed. 
 
Thus, every term of $\Sig_\exc$ and $\Sig_\core$ can be interpreted 
by first converting the interpretation of each of its components $f:X\to Y$ 
to a function $\deno{f}:\deno{X}+\Exc \to \deno{Y}+\Exc$. 
 
For $\Sig_\exc$, this coincides with the \emph{Kleisli composition} associated 
to the \emph{exception monad} $A+\Exc$~\cite{Mo91}. 

We will also use the following conversion. 
 
\begin{definition}
\label{defi:downcast}
The \emph{downcasting} conversion is the following transformation:
\begin{itemize}
\item every function $\theta:A+\Exc \to B+\Exc$ gives rise to 
$ \downcast{\theta}=\theta\circ\inn_{A}:A\to B+\Exc $
which is equal to $\theta$ on $A$ and which propagates exceptions. 
\end{itemize}
\end{definition}
This conversion is unsafe: different $\theta$'s may
give rise to the same $\downcast{\theta}$.

\subsection{Tagging and raising exceptions}\label{subsec:den-tag}

Raising exceptions relies on the interpretation of the tagging operations.
The interpretation of the empty type $\empt$ is the empty set 
$\emptyset$; thus, for each type $X$ the interpretation of $\empt+X$ 
can be identified with $\deno{X}$.
\begin{definition}
\label{defi:den-tag} 
Let $\Exc$ be the disjoint union of the sets $\deno{T}$ 
for all the exceptional types~$T$. 
Then, for each exceptional type $T$, the interpretation of the tagging 
operation $\tagg{T}:T\to\empt$ is the 
coprojection function 
  $$\deno{\tagg{T}}:\deno{T}\to\Exc \;.$$
\end{definition}
 
Thus, the tagging function $\deno{\tagg{T}}:\deno{T}\to\Exc$
maps a non-exceptional value (or \emph{parameter}) $a\in\deno{T}$  
to an exception $\deno{\tagg{T}}(a)\in\Exc$.

We can now define the raising of exceptions 
in a programming language. 
\begin{definition}
\label{defi:den-raise}
For each exceptional type $T$ and each type $Y$, 
the interpretation of the raising operation $ \throw{T}{Y} $ 
is the tagging function $\deno{\tagg{T}}$ followed by 
the inclusion of $\Exc$ in $\deno{Y}+\Exc$: 
  $$ \deno{\throw{T}{Y}} = \ina_{\deno{Y}} \circ \deno{\tagg{T}} : 
  \deno{T} \to \deno{Y}+\Exc \;.$$
\end{definition}

\subsection{Untagging and handling exceptions}\label{subsec:den-untag}

Handling exceptions relies on the interpretation of the untagging operations 
for clearing the exception tags. 
\begin{definition}
\label{defi:den-untag}
For each exceptional type $T$, the interpretation of the untagging 
operation $\untag{T}:\empt\to T$ is the function
  $$ \deno{\untag{T}}:\Exc\to \deno{T}+\Exc \;, $$
which satisfies for each exceptional type $R$:

$$
\deno{\untag{T}}\circ\deno{\tagg{R}} = 
\begin{cases}
\inn_{\deno{T}} & \text{~when~} R= T \\
\ina_{\deno{T}}\circ\deno{\tagg{R}} & \text{~when~} R\ne T \\ 
\end{cases}
  \;\; : \deno{R}\to \deno{T}+\Exc. 
$$
\end{definition}
 
Thus, the untagging function $\deno{\untag{T}}$,
when applied to any exception $e$, 
first tests whether $e$ is in $\deno{T}$; 
if this is the case, 
then it returns the parameter $a\in \deno{T}$ such that $e=\deno{\tagg{T}}(a)$, 
otherwise it propagates the exception~$e$. 
 
Since the domain of $\deno{\untag{T}}$ is $\Exc$, 
$\deno{\untag{T}}$ is uniquely determined by its restrictions 
to all the exceptional types, 
and therefore by the equalities in Definition~\ref{defi:den-untag}.

For handling exceptions of types $T_1,\ldots T_n$, raised by the interpretation
of some term $f:X\to Y$ of $\Sig_\exc$, one provides for each $i$ in
$\{1,\dots,n\}$ a term $g_i:T_i \to Y$ of $\Sig_\exc$ 
(thus, the interpretation of $g_i$ may itself raise exceptions). 
Then the handling process builds a function 
which first executes $f$, 
and if $f$ returns an exception then maybe catches this exception. 
The catching part encapsulates some untagging functions,
but the resulting function always propagates exceptions. 
\begin{definition}
\label{defi:den-handle}
For each term $f:X\to Y$ of $\Sig_\exc$, 
and each non-empty lists $(T_i)_{1\leq i\leq n}$ of exceptional types
and $(g_i:T_i\to Y)_{1\leq i\leq n}$ of terms of $\Sig_\exc$, 
let $(\recov_i)_{1\leq i\leq n}$ 
denote the family of functions defined recursively by: 
  $$ \recov_i \;=\; \begin{cases} 
    \cotuple{\; \deno{g_n} \;|\; \ina_Y \;} \circ \untag{T_n} & 
       \mbox{ when } i=n \\
    \cotuple{\; \deno{g_i} \;|\; \recov_{i+1} \;} \circ \untag{T_i} & 
       \mbox{ when } i< n \\
  \end{cases} 
  \quad : \Exc \to Y+\Exc $$

Then the interpretation of the handling operation is:
$$\deno{ \try{f}{\catchi}} =  
\cotuple{\; \inn_Y \;|\; \recov_1 \;} \circ \deno{f}.$$
\end{definition}
It should be noted that $\deno{f}:\deno{X}\to \deno{Y}+\Exc $ 
and that similarly $\deno{ \try{f}{\catchi}} : \deno{X}\to \deno{Y}+\Exc $. 
 
When $n=1$ we get: 
$$\deno{ \try{f}{\catch{T}{g}} } =  
  \cotuple{\; \inn_Y \;|\; \cotuple{\; \deno{g} \;|\; \ina_Y \;} 
  \circ \untag{T} \;} \circ \deno{f} \;.$$

This definition matches that of Java exceptions~\cite[Ch.~14]{java} 
or C++ exceptions~\cite[\S 15]{cpp}.

In particular, in the interpretation of $\try{f}{\catchi}$, 
each function $\deno{g_i}$ may itself raise exceptions; 
 
and the types $T_1,\dots,T_n$ need not be pairwise distinct,
but if $T_i=T_j$ for some $i<j$ then $g_j$ is never executed.

\section{A decorated equational proof system for exceptions}
 
\label{sec:deco}

In Sections~\ref{sec:syn} and~\ref{sec:den} we have formalized 
the signature for exceptions $\Sig_{\exc}$, its associated 
core signature $\Sig_{\core}$, 
and we have described their denotational semantics.
However the soundness property is not satisfied, in the sense that 
the denotational semantics is not a model of the signature,
in the usual sense: 
indeed, a term $f:X \to Y$ is not always interpreted 
as a function $\deno{f}:\deno{X}\to \deno{Y}$;
it may be interpreted as $\deno{f}:\deno{X} \to \deno{Y}+\Exc$, 
or as $\deno{f}:\deno{X}+\Exc \to \deno{Y}+\Exc$.
 
In order to get soundness, in this Section 
we add \emph{decorations} to the signature for exceptions 
by classifying the operations and equations 
according to the interaction of their interpretations 
with the mechanism of exceptions.
$$
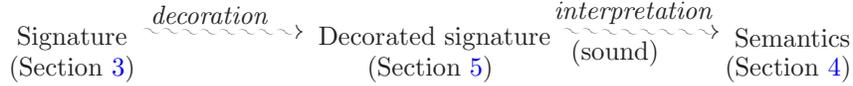
\xymatrix@C=5pc@R=-.3pc{
\txt{Signature } \ar@{~>}[r]^(.4){\txt{\it decoration}\;\;} & 
\txt{ Decorated signature } 
  \ar@{~>}[r]^(.55){\;\;\txt{\it interpretation}}_{\;\;\txt{(sound)}} & 
\txt{ Semantics} \\ 
\txt{(Section~\ref{sec:syn})} &
\txt{(Section~\ref{sec:deco})} &
\txt{(Section~\ref{sec:den})} \\
}$$

\subsection{Decorations for exceptions}\label{subsec:deco-deco}

By looking at the interpretation (in Section~\ref{sec:den})
of the syntax for exceptions (from Section~\ref{sec:syn}),  
we get a classification of the operations and terms in three parts, 
depending on their interaction with the exceptions mechanism.
 
The terms are decorated by $\pure$, $\ppg$ and $\ctc$ used as
superscripts, they are called respectively 
\emph{pure} terms, \emph{propagators} and \emph{catchers},
according to their interpretations:
 \begin{description}
\item[$\pure$] the interpretation of a \emph{pure} term may neither raise
  exceptions nor recover from exceptions,
\item[$\ppg$] the interpretation of a \emph{propagator} may raise exceptions but
  is not allowed to recover from exceptions, 
\item[$\ctc$] the interpretation of a \emph{catcher} may raise exceptions and
  recover from exceptions.
\end{description}
For instance, the decoration~$\pure$ corresponds to the decoration
\texttt{noexcept} in C++ (replacement of the deprecated
\texttt{throw()}) and the decoration~$\ppg$ corresponds to
\texttt{throw(...)}, still in C++. 
The decoration~$\ctc$ is usually  {\em not} encountered in the
language, since catching is the prerogative of the \emph{core}
untagging function, which is private.

Similarly, we introduce two kinds of equations between terms.
This is done by using two distinct relational symbols
which correspond to two distinct interpretations:
\begin{description}
\item[$(\eqs)$] a \emph{strong} equation is an equality of functions 
both on ordinary values and on exceptions
\item[$(\eqw)$]  a \emph{weak} equation is an equality of functions 
only on ordinary values, but maybe {\em not on exceptions}. 
\end{description}
 
The interpretation of these three kinds of terms and two kinds of equations
is summarized in Fig.~\ref{fig:expansion}.

It has been shown in Section~\ref{subsec:den-sem} that 
any propagator can be seen as a catcher 
and that any pure term can be seen as a propagator and thus also as a catcher.
This allows to compose terms of different nature,

so that it is not a restriction to give the interpretation 
of the decorated equations only when both members are catchers.

\begin{figure}[!ht]
\renewcommand{\arraystretch}{1.2}   
$$\begin{array}{|c|c|c|} 
\hline 
\mbox{Syntax} & 
\mbox{Decorated syntax} & 
\mbox{Interpretation}  \\ 
\hline 
\xymatrix@R=.5pc{\mbox{type}\\} & 
\xymatrix@R=.5pc{\mbox{type}\\} & \\
\xymatrix@R=.5pc{ X \\ } & 
\xymatrix@R=.5pc{ X \\ } & 
\xymatrix@R=.5pc{ \deno{X} \\ } \\
\hline 
\xymatrix@R=.5pc{\mbox{term}\\} & 
\xymatrix@R=.5pc{\mbox{pure term}\\} & \\
\xymatrix@R=.5pc{ X \ar[r]^{f} & Y \\ } & 
\xymatrix@R=.5pc{ X \ar[r]^{f^\pure} & Y \\ } & 
\xymatrix@R=.5pc@C=4pc{
  \deno{X} \ar[r]^{\deno{f}} & \deno{Y} \\ } \\
\hline 
\xymatrix@R=.5pc{\mbox{term}\\} & 
\xymatrix@R=.5pc{\mbox{propagator}\\} & \\
\xymatrix@R=.5pc{ X \ar[r]^{f} & Y \\ } & 
\xymatrix@R=.5pc{ X \ar[r]^{f^\ppg} & Y \\ } & 
\xymatrix@R=.5pc@C=4pc{
  \deno{X} \ar[r]^(.45){\deno{f}} & \deno{Y}+\Exc \\ } \\
\hline 
\xymatrix@R=.5pc{\mbox{term}\\} & 
\xymatrix@R=.5pc{\mbox{catcher}\\} & \\
\xymatrix@R=.5pc{ X \ar[r]^{f} & Y \\ } & 
\xymatrix@R=.5pc{ X \ar[r]^{f^\ctc} & Y \\ } & 
\xymatrix@R=.5pc@C=4pc{ 
  \deno{X}+\Exc \ar[r]^{\deno{f}} & \deno{Y}+\Exc \\ } \\
\hline 
\mbox{equation} & 
\mbox{strong equation} & \\
\begin{array}{l} f = g : X \to Y \end{array} &
\begin{array}{l} f^\ctc \eqs g^\ctc : X \to Y \end{array} &
\;\; \deno{f} = \deno{g} \\
\hline 
\mbox{equation} & 
\mbox{weak equation} & \\
\begin{array}{l} f = g : X \to Y \end{array} &
\begin{array}{l} f^\ctc \eqw g^\ctc :  X \to Y \end{array} &
\;\; \deno{f}\;\circ\;\inn_{\deno{X}} = \deno{g}\;\circ\;\inn_{\deno{X}} \\
\hline 
\end{array}$$
\renewcommand{\arraystretch}{1}
\caption{Interpretation of the decorated syntax.} 
\label{fig:expansion} 
\end{figure}

Now we can add decorations to the signature for exceptions $\Sig_\exc$
and its associated core signature $\Sig_\core$,
from Definitions~\ref{defi:syn-sig} and~\ref{defi:syn-core}.
 
\begin{definition}
\label{defi:deco-sig}

The \emph{decorated signature for exceptions} $\Sig_\exc^\deco$ 
and its  associated \emph{decorated core signature} $\Sig_\core^\deco$
are made of $\Sig_\exc$ and $\Sig_\core$, respectively, decorated as follows: 
the basic operations are pure, 
the tagging, raising and handling operations are propagators,
and the untagging operations are catchers. 
\end{definition}

\subsection{Decorated rules for exceptions}\label{subsec:deco-rules}

In this Section we define an equational proof system for exceptions. 
This proof system is made of the rules in Fig.~\ref{fig:rules}.
It can be used for

proving properties of exceptions, for instance in the Coq proof assistant.
 
In Fig.~\ref{fig:rules}, 
the decoration properties are often grouped
with other properties: for instance, ``$f^\ppg \eqw g^\ppg$''
means ``$f^\ppg$ and $g^\ppg$ and $f \eqw g$''; 
in addition, the decoration $\ctc$ is usually dropped, since the rules
assert that every term can be seen as a catcher;
and several rules with the same premisses may be grouped together:
$\frac{H_1\dots H_n}{C_1\dots C_p}$ stands for 
$\frac{H_1\dots H_n}{C_1}$,\dots,$\frac{H_1\dots H_n}{C_p}$.

\begin{figure}[H]  \center 
\renewcommand{\arraystretch}{2.3}   
$$ \begin{array}{|c|} 
\hline 
\multicolumn{1}{|l|}
  {\text{\firstmonadrule Monadic equational rules for exceptions (first part):} } \\ 
 
\; \dfrac{f:X\to Y \squad g:Y\to Z}{g\circ f:X\to Z} \quad 
 
\dfrac{X}{\id_X:X\to X } \quad 
 
\dfrac{f}{f \eqs f} \quad 
\dfrac{f \eqs g}{g \eqs f} \quad 
\dfrac{f \eqs g \squad g \eqs h}{f \eqs h} \; \\
 
\dfrac{f:X\to Y \quad g_1\eqs g_2:Y\to Z}
  {g_1\circ f \eqs g_2\circ f :X\to Z}  \qquad
\dfrac{f_1\eqs f_2:X\to Y \quad g:Y\to Z}
  {g\circ f_1 \eqs g\circ f_2 :X\to Z} \\ 
 
\dfrac{f:X\to Y \quad g:Y\to Z \quad h:Z\to W}
  {h\circ (g\circ f) \eqs (h\circ g)\circ f} \qquad 
\dfrac{f:X\to Y}{f\circ \id_X \eqs f} \qquad 
\dfrac{f:X\to Y}{\id_Y\circ f \eqs f} \\ 
\hline 
\multicolumn{1}{|l|}
  {\text{\seconmonadrule Monadic equational rules for exceptions (second part):} } \\ 
 
\;\; \dfrac{f^\pure}{f^\ppg} \quad 
 
\dfrac{f^\ppg}{f^\ctc} \quad 
 
\dfrac{X}{\id_X^\pure } \qquad 
 
\dfrac{f^\pure \quad g^\pure}{(g\circ f)^\pure}  \quad 
 
\dfrac{f^\ppg \quad g^\ppg}{(g\circ f)^\ppg} \qquad 
 
\dfrac{f^\ppg \eqw g^\ppg}{f \eqs g} \quad
 
\dfrac{f \eqs g}{f \eqw g} \;\; \\
 
\dfrac{f}{f \eqw f} \qquad  
\dfrac{f \eqw g}{g \eqw f} \qquad 
\dfrac{f \eqw g \quad g \eqw h}{f \eqw h}  \\
 
\dfrac{f^\pure:X\to Y \quad g_1\eqw g_2:Y\to Z}
  {g_1\circ f \eqw g_2\circ f }  \qquad
\dfrac{f_1\eqw f_2:X\to Y \quad g:Y\to Z}
  {g\circ f_1 \eqw g\circ f_2 } \\
\hline 
\multicolumn{1}{|l|}
  {\text{\emptyrule Rules for the empty type $\empt$:} \qquad
 
\dfrac{X}{\cotu_X^\pure:\empt \to X } \qquad 
 
\dfrac{f,g:\empt \to Y}{f \eqw g} \; }  \\
\hline 
\multicolumn{1}{|l|}
  {\text{\caserule Case distinction with respect to $X+\empt$:}  } \\ 
 
\;\dfrac{g^\ppg\!:\!X\!\to\! Y \;\; k^\ctc\!:\!\empt\!\to\! Y}
  {\cotuple{g\,|\,k}^\ctc\!:\!X\to Y \quad
  \cotuple{g\,|\,k} \eqw g \quad 
  \cotuple{g\,|\,k} \circ \cotu_X \eqs k } \; \\ 
 
\dfrac{f,g : X\to Y \quad f \eqw g \quad 
  f \circ \cotu_X \eqs g \circ \cotu_X}{f\eqs g} \\
\hline 
\multicolumn{1}{|l|}
  {\text{\propagrule Propagating exceptions:}  \qquad  
 
  \dfrac{k^\ctc:X\to Y}{(\downcast{k})^\ppg:X\to Y \quad  
  \downcast{k} \eqw k}}  \\ 
\hline 
\multicolumn{1}{|l|}
{\text{\tagrule Tagging:} }  \\

\dfrac {T\in\exctype} 
  {\tagg{T}^\ppg : T\to\empt}  \qquad\qquad  
 
\dfrac {(f_T^\ppg : T\to Y)_{T\in\exctype}} 
  {\cotuple{f_T}_{T\in\exctype}^\ctc : \empt\to Y \quad   
  \cotuple{f_T}_{T\in\exctype} \circ \tagg{T} \eqw f_T } \\
 
\dfrac{f,g : \empt \to Y \quad 
  f \circ \tagg{T} \eqw g \circ \tagg{T} \mbox{ for all } T\in\exctype \; }
  {f \eqs g }  \\
\hline 
\multicolumn{1}{|l|}
{\text{\untagrule Untagging:} }  \\ 
\dfrac {T\in\exctype}{\untagd{T}{\ctc}: \empt\to T  \quad 
\untag{T}\circ\tagg{T} \eqw \id_T} \qquad 
\dfrac {R,T\in\exctype \quad R\ne T}
{\untag{T}\circ\tagg{R} \eqw \cotu_T \circ\tagg{R} }\\
\hline 
\end{array}$$
\renewcommand{\arraystretch}{1}
\caption{Decorated rules for exceptions} 
\label{fig:rules} 
\end{figure}

Rules \firstmonadrule are the usual rules for 
the monadic equational logic; they are valid for all decorated terms 
and for strong equations. 
 
Rules \seconmonadrule provide more information on 
the decorated monadic equational logic for exceptions; 
in particular, the substitution of $f$ in a weak equation 
$g_1\eqw g_2$ holds only when $f$ is pure, which is quite restrictive.  
 
Rules \emptyrule ensure that the empty type is a kind of initial
type with respect to weak equations. 
 
Rules \caserule are used for case distinctions between 
exceptional arguments (on the ``$\empt$'' side of $X+\empt$)
and non-exceptional arguments (on the ``$X$'' side of $X+\empt$). 
 
The symbol $\downcast{}$ in rules \propagrule is interpreted 
as the downcast conversion, see Definition~\ref{defi:downcast}.
 
Rules in \tagrule mean that the tagging operations are interpreted as 
the canonical inclusions of the exceptional types in the set $\Exc$, 
see Definition~\ref{defi:den-tag}. 
 
The rules in \untagrule determine the untagging operations 
up to strong equations:  
an untagging operation recovers the exception parameter 
whenever the exception type is matched, 
and it propagates the exception otherwise 
see Definition~\ref{defi:den-untag}. 

\begin{remark} 
\label{rem:dual}
It has been shown in \cite{DDFR12-dual}
that the denotational semantics of the core language for 
exceptions is dual to the denotational semantics for states: 
the \emph{tagging} and \emph{untagging} operations 
are respectively dual to the \emph{lookup} and \emph{update} operations. 
In fact, this duality is also valid for the decorated equational logics.
\end{remark}

This decorated proof system is used now (in Definition~\ref{defi:deco-handle})
for constructing the raising and handling operations 
from the core tagging and untagging operations.
It has to be noted that the term $\catchi\circ f$ may catch exceptions, 
while the handling operation, which coincides with $\catchi\circ f$
on non-exceptional values, must propagate exceptions; 
this is why the downcast operator $\downcast{}$ is used.

\begin{definition}
\label{defi:deco-raise}
For each exceptional type $T$ and each type $Y$, 
the \emph{raising} operation $ \throwd{T}{Y}{\ppg}:T\to Y$ 
is the propagator defined as:
	$$ \throwd{T}{Y}{\ppg} = \cotu_Y \circ\tagg{T} \;.$$ 
\end{definition}

\begin{definition}
\label{defi:deco-handle}

For each propagator $f^\ppg:X\to Y$, 
  each non-empty lists of types $(T_i)_{1\leq i\leq n}$ 
  and propagators $(g_i^\ppg:T_i\to Y)_{1\leq i\leq n}$, 
let $\catchi^\ctc : Y \to Y $ denote the catcher defined by: 
$$ \catchi^\ctc = \cotuple{\; \id_Y \;|\; \recov_1 \;}
$$
where $(\recov_i^\ctc: \empt \to Y)_{1\leq i\leq n}$ 
denotes the family of catchers defined recursively by: 
  $$ \recov_i^\ctc \;=\; \begin{cases} 
    g_n \circ \untag{T_n} & 
       \mbox{ when } i=n, \\
    \cotuple{g_i \;|\; \recov_{i+1} \;} \circ \untag{T_i} & 
       \mbox{ when } i< n. \\
  \end{cases} $$
Then, the \emph{handling} operation $ (\try{f}{\catchi})^\ppg :X\to Y$ 
is the propagator defined by: 
$$ \try{f}{\catchi} = \downcast{ \left(\catchi\circ f\right)}$$
\end{definition}
 
When $n=1$ we get:
$$ \try{f}{\catch{T}{g}} =  
  \downcast{\, (\catch{T}{g}\circ f)} = 
  \downcast{\,( \cotuple{\; \id_Y \;|\; g \circ \untag{T} \;}\circ f )}.
$$

Now Theorem~\ref{theo:rules-sound} derives easily, by induction, 
from Fig.~\ref{fig:rules} and~\ref{fig:expansion} and from 
Definitions~\ref{defi:deco-raise} and~\ref{defi:deco-handle}. 

\begin{theorem}\label{theo:rules-sound} 
The decorated rules for exceptions and 
the raising and handling constructions 
are sound with respect to the denotational semantics of exceptions. 
\end{theorem}

\subsection{A decorated proof: a propagator propagates}\label{ssec:lemma}

With these tools, 
it is now possible to prove properties of programs involving
exceptions and to check these proofs in Coq.
For instance, let us {\em prove} that given an exception, 
a propagator will do nothing apart from propagating it. 
Recall that the interpretation of $\cotu_Z$ 
(or, more precisely, of $\upcastboth{\cotu_Z}$) is 
$\ina_{\deno{Z}}:\Exc\to\deno{Z}+\Exc$.

\begin{lemma}
\label{lem:coprod-cotu} 
For each propagator $g^\ppg:X\to Y$ we have 
$g\circ \cotu_X \eqs \cotu_Y$. 
\end{lemma}
\begin{proof}
This lemma can be proved as follows; 
the labels refer to Fig.~\ref{fig:rules}, and their subscripts to the
proof in Coq of Fig.~\ref{fig:propagpropag}.
\small
\begin{prooftree}
\AxiomC{$X$}
\LeftLabel{$\emptyrule_1\;$} 
\UnaryInfC{$\cotu_X : \empt \to X$}
    \AxiomC{$g:X\to Y$}
  \LeftLabel{$\firstmonadrule\;$} 
  \BinaryInfC{$g\circ \cotu_X : \empt \to Y$}
  \LeftLabel{$\emptyrule_2\;$} 
  \UnaryInfC{$g\circ \cotu_X \eqw \cotu_Y$}
         \AxiomC{$g^\ppg$}    
                      \AxiomC{$X$}    
                      \LeftLabel{$\emptyrule_3\;$} 
                      \UnaryInfC{$\cotu_X^\pure$}
                      \LeftLabel{$\seconmonadrule_1\;$} 
                      \UnaryInfC{$\cotu_X^\ppg$}
                \LeftLabel{$\seconmonadrule_2\;$} 
                \BinaryInfC{$(g\circ \cotu_X)^\ppg$}
                            \AxiomC{$Y$}    
                            \LeftLabel{$\emptyrule_4\;$} 
                            \UnaryInfC{$\cotu_Y^\pure$}
                            \LeftLabel{$\seconmonadrule_3\;$} 
                            \UnaryInfC{$\cotu_Y^\ppg$}
      \LeftLabel{$\seconmonadrule_4\;$} 
      \TrinaryInfC{$g\circ \cotu_X \eqs \cotu_Y$}
\end{prooftree}
\normalsize
\end{proof}
 
The proof in Coq follows the same line as the mathematical proof above. 
It goes as in Figure~\ref{fig:propagpropag}.
The Coq library for exceptions, 
\href{https://forge.imag.fr/frs/download.php/541/EXCEPTIONS-0.1.tar.gz}{EXCEPTIONS-0.1.tar.gz}, can be
found online: \href{https://forge.imag.fr/frs/?group_id=506}{http://coqeffects.forge.imag.fr} 
(in file \texttt{Proofs.v}) 
with proofs of many other properties of programs involving exceptions. 
\begin{figure}[ht]
\begin{jgdfrsh}\input{propagpropag.v.tex}\end{jgdfrsh}
\caption{Proof in Coq that ``a propagator propagates exceptions''}
\label{fig:propagpropag}
\end{figure}

It should be recalled that any Coq proof is read from bottom up. 
Last, the application of the 
\texttt{from\_empty\_is\_weakly\_unique} rule certifies  
that the term \texttt{g o (@empty X)}, because its domain is the empty type, 
is weakly equal to the term \texttt{(@empty Y)}. 
Thus, we have the left side sub-proof: 
\texttt{g o (@empty X) $\sim$ (@empty Y)}. 
There, weak equality is converted into strong:
applying \texttt{propagator\_weakeq\_is\_strongeq} 
resolves the goal: \texttt{g o (@empty X) == (@empty Y)} and produces
as sub-goals that there is no catcher involved in both hand sides  
(middle and right side sub-proofs).
 
\subsection{A hierarchy of exceptional types}\label{ssec:hierarchy}

In object-oriented languages, exceptions are usually the objects of classes 
which are related by a hierarchy of subclasses.
Our framework can be extended in this direction by introducing 
a hierarchy of exceptional types:
the set $\exctype$ is endowed with a partial order $\subtype$ called 
the \emph{subtyping relation},
and the signature $\Sig_{\basic}$ is extended with 
a \emph{cast} operation $\cast{R}{T}:R\to T$ whenever $R \subtype T$.
 
The interpretation of $\cast{R}{T}$ is a pure function 
$\deno{\cast{R}{T}}:\deno{R}\to \deno{T}$,
such that 
$\deno{\cast{T}{T}}$ is the identity on $\deno{T}$ 
and when $S\subtype R\subtype T$ then 
$\deno{\cast{S}{T}}=\deno{\cast{R}{T}}\circ\deno{\cast{S}{R}}$.
 
Definition~\ref{defi:den-untag} has to be modified as follows: 
the function $ \deno{\untag{T}}:\Exc\to \deno{T}+\Exc $ 
satisfies for each exceptional type $R$: 
$$
\deno{\untag{T}}\circ\deno{\tagg{R}} = 
\begin{cases}
\inn_{\deno{T}}\circ\deno{\cast{R}{T}} & \text{~when~} R\subtype T \\
\ina_{\deno{T}}\circ\deno{\tagg{R}} & \text{~otherwise}  \\ 
\end{cases}
  \;\; : \deno{R}\to \deno{T}+\Exc. 
$$

\section{Conclusion}
\label{sec:conc}

Exceptions are part of most modern programming languages 
and they are useful in computer algebra, 
typically for implementing dynamic evaluation. 
 
We have presented a new framework for formalizing the treatment of
exceptions. 
These decorations form a bridge between the syntax and the denotational
semantics by turning the syntax sound with respect to the semantics, 
without adding any explicit ``type of exceptions''
as a possible return type for the operations. 
 
The salient features of our approach are:
\begin{itemize}
\item We provide rules for equational proofs on programs 
involving exceptions (Fig.~\ref{fig:rules}) and 
an automatic process for interpreting these proofs (Fig.~\ref{fig:expansion}). 
\item Decorating the equations allows to separate properties 
that are true only up to effects (weak equations) from 
properties that are true even when effects are considered (strong equations).

\item Moreover,
the verification of the proofs can be done in two steps: in a first
step, decorations are dropped and the proof is checked syntactically; in a
second step, the decorations are taken into account in order to prove
properties involving computational effects. 
\item The distinction between the 
language for exceptions and its associated private core language
(Definitions~\ref{defi:syn-sig} and~\ref{defi:syn-core})
allows to split the proofs in two successive parts;
in addition, 
the private part can be directly dualized from the proofs on global states
(relying on \cite{DDFR12-dual} and \cite{DDFR12-state}).

\item A proof assistant can be used for checking the decorated proofs on
exceptions. Indeed the decorated proof system for states, as described
in~\cite{DDFR12-state} has been implemented in Coq~\cite{DDEP13-coq}
and dualized for exceptions (see \url{http://coqeffects.forge.imag.fr}). 
 
\end{itemize}

We have used the approach of decorated logic, which provides 
rules for computational effects by adding decorations to 
usual rules for ``pure'' programs. 
This approach can be extended in order to deal with 
multivariate operations~\cite{DDR11-seqprod}, 
conditionals, loops, and high-order languages.

\end{document}